\documentclass[a4paper]{article}
\pdfoutput=1
\usepackage[hypertexnames=false]{hyperref} 
\usepackage[svgnames]{xcolor} 
\usepackage[T1]{fontenc}
\usepackage{lmodern}
\usepackage{a4wide}

\usepackage{color}
\usepackage[protrusion=true,expansion=true]{microtype}
\usepackage[utf8]{inputenc}
\usepackage{amsfonts, amsopn,amsmath, amsthm}
\usepackage{graphicx}
\graphicspath{{./}{../}}
\usepackage{xspace}	
\usepackage{paralist}
\usepackage{authblk} 
\usepackage[linesnumbered, vlined]{algorithm2e}
\usepackage[absolute]{textpos}
\usepackage{tikz}
\usetikzlibrary{calc,positioning}

\usepackage{tabularx}

\providecommand{\keywords}[1]{\bigskip\textbf{\textit{Keywords:}} #1}

\newtheorem{theorem}{Theorem}
\newtheorem{Theorem}[theorem]{Theorem}
\newtheorem{Proposition}{Proposition}
\newtheorem{Lemma}[theorem]{Lemma}

\newtheorem{Corollary}[theorem]{Corollary}

\newcommand{\Trule}{\rule{0pt}{3ex}}
\newcommand{\Brule}{\rule[-1.5ex]{0pt}{0pt}}

\newcommand{\optprobA}[3]{
\begin{center}
\begin{tabularx}{0.8\textwidth}{|l X|}
	\hline
	\multicolumn{2}{|c|}{#1\Trule} \\
	\textbf{Input:\ }&{#2}\\
	\textbf{Solution:\ }&{#3\Brule}\\
	\hline
\end{tabularx}
\end{center}
}

\newcommand{\parprobA}[4]{
\begin{center}
\begin{tabularx}{0.8\textwidth}{|l X|}
	\hline
	\multicolumn{2}{|c|}{#1\Trule} \\
	\textbf{Input:\ }&{#2}\\
	\textbf{Parameter:\ }&{#3}\\
	\textbf{Solution:\ }&{#4\Brule}\\
	\hline
\end{tabularx}
\end{center}
}

\newcommand{\FPT}{\ensuremath{\mathsf{FPT}}\xspace}

\newcommand{\W}[1]{\ensuremath{\mathsf{W[#1]}}\xspace}
\newcommand{\NP}{\ensuremath{\mathsf{NP}}\xspace}

\newcommand{\F}{\ensuremath{\mathcal{F}}\xspace}

\newcommand{\sumdiversity}{\textrm{div}_{\mathrm{total}}}
\newcommand{\mindiversity}{\textrm{div}_{\mathrm{min}}}
\newcommand{\Hammingdistance}{d_H}

\newcommand{\df}{:=}

\newcommand{\Ss}{\ensuremath{\mathcal{S}}\xspace}

\newcommand{\Ocal}{\ensuremath{\mathcal{O}}\xspace}

\newcommand{\bR}{\mathbb{R}}

\newcommand{\sol}{\texttt{sol}}

\newcommand{\VC}{\textsc{Vertex Cover}\xspace}
\newcommand{\dHS}{\textsc{$d$-Hitting Set}\xspace}
\newcommand{\FVS}{\textsc{Feedback Vertex Set}\xspace}
\newcommand{\fvs}{feedback vertex set}
\newcommand{\rDkVC}{\textsc{Diverse Vertex Cover}\xspace}
\newcommand{\rDkdHS}{\textsc{Diverse $d$-Hitting Set}\xspace}
\newcommand{\rDkFVS}{\textsc{Diverse Feedback Vertex Set}\xspace}
\newcommand{\rMDkVC}{\textsc{Min-Diverse Vertex Cover}\xspace}
\newcommand{\rMDkdHS}{\textsc{Min-Diverse $d$-Hitting Set}\xspace}
\newcommand{\rMDkFVS}{\textsc{Min-Diverse Feedback Vertex Set}\xspace}

\newcommand{\maxcostflow}{\textsc{Maximum Cost Flow}\xspace}
\newcommand{\MWBM}{\textsc{Maximum Weight $b$-Matching}\xspace}

\renewcommand\paragraph[1]{\medskip\noindent\textbf{#1}}

\newenvironment{question}[1]
	{~\par\medskip\paragraph{Open Question (#1).}}
	{~\par\medskip}

\title{
  FPT Algorithms for
  Diverse Collections of Hitting Sets%
  \thanks{Tomáš Masařík received funding from the European Research Council
    (ERC) under the European Union's Horizon 2020 research and innovation
    programme Grant Agreement 714704, and from Charles University student grant
    SVV-2017-260452. Lars Jaffke is supported by the Bergen Research Foundation
    (BFS). Geevarghese Philip received funding from the following sources: the
    European Research Council (ERC) under the European Union's Horizon 2020
    research and innovation programme (Grant Agreement No 819416), the Norwegian
    Research Council via grants MULTIVAL and CLASSIS, BFS (Bergens Forsknings
    Stiftelse) "Putting Algorithms Into Practice" Grant Number 810564 and NFR
    (Norwegian Research Foundation) grant number 274526d "Parameterized
    Complexity for Practical Computing".}
    \thanks{The final version of this article has been accepted in the Special Issue \textit{New Frontiers in Parameterized Complexity and Algorithms} of Algorithms journal~\cite{final}.}
  }

\author[1]{Julien Baste}
\author[2]{Lars Jaffke}
\author[3,4]{Tomáš Masařík}
\author[5]{\\ Geevarghese Philip}
\author[6]{G\"unter Rote}

\affil[1]{Institute of Optimization and Operations Research, Ulm University, Germany}
\affil[ ]{\texttt{julien.baste@uni-ulm.de}}
\affil[2]{University of Bergen, Norway}
\affil[ ]{\texttt{lars.jaffke@uib.no}}
\affil[3]{Charles University, Prague, Czech Republic}
\affil[4]{University of Warsaw, Poland}
\affil[ ]{\texttt{masarik@kam.mff.cuni.cz}}
\affil[5]{Chennai Mathematical Institute, Chennai, India and UMI ReLaX}
\affil[ ]{\texttt{gphilip@cmi.ac.in}}
\affil[6]{Freie Universit\"at Berlin}
\affil[ ]{\texttt{rote@inf.fu-berlin.de}}
\date{}

\begin{document}

\maketitle

\begin{textblock}{20}(0, 12.5)
\includegraphics[width=40px]{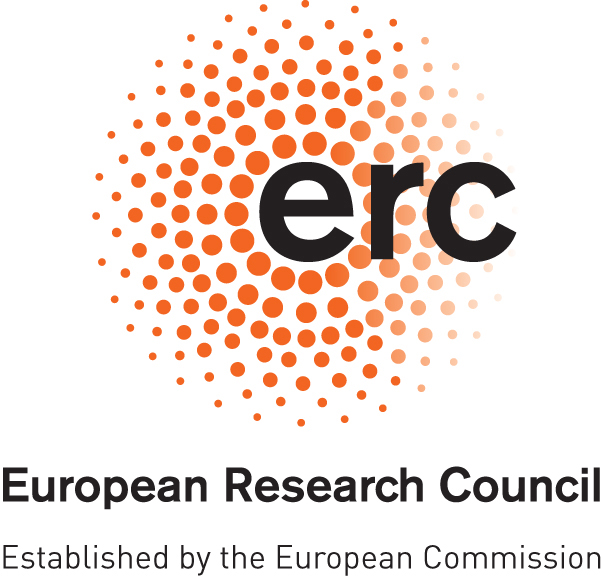}%
\end{textblock}
\begin{textblock}{20}(-0.25, 12.9)
\includegraphics[width=60px]{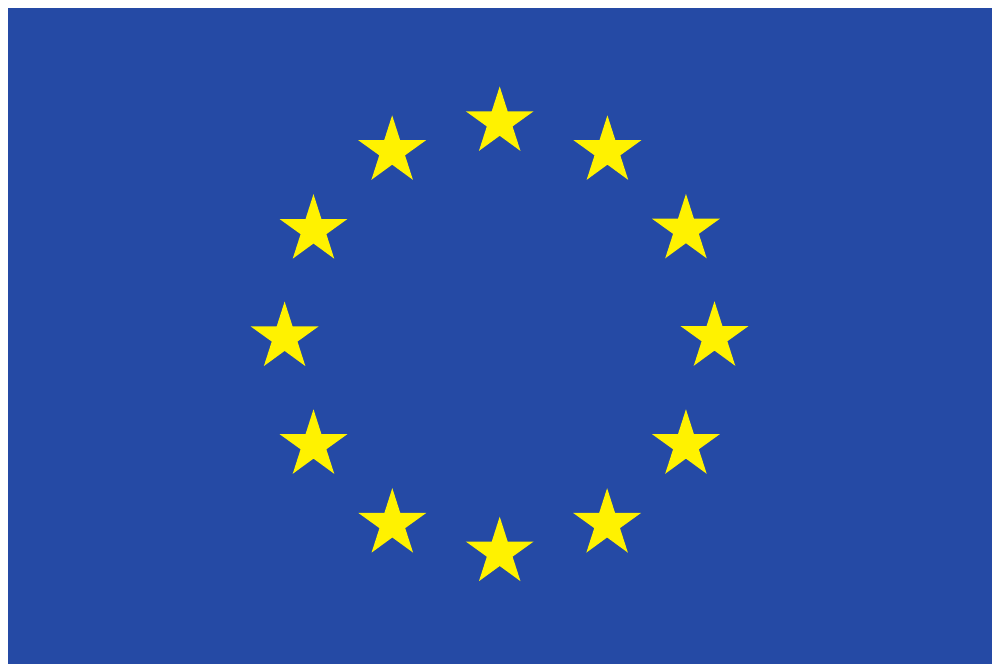}%
\end{textblock}

\abstract{In this work, we study the \dHS and \FVS problems through the paradigm of
  finding diverse collections of $r$ solutions of size at most $k$ each, which
  has recently been introduced to the field of parameterized complexity [Baste
  et al., 2019]. This paradigm is aimed at addressing the loss of important
  \emph{side information} which typically occurs during the abstraction process
  which models real-world problems as computational problems. We use two
  measures for the diversity of such a collection: the sum of all pairwise
  Hamming distances, and the minimum pairwise Hamming distance. We show that
  both problems are FPT in $k + r$ for both diversity measures. A key ingredient
  in our algorithms is a (problem independent) network flow formulation that,
  given a set of `base' solutions, computes a maximally diverse collection of
  solutions. We believe that this could be of independent interest.}

\keywords{Solution Diversity; Fixed-Parameter Tractability; Hitting Sets; Vertex Cover; Feedback Vertex Set; Hamming Distance}

\section{Introduction}
The typical approach in modeling a real-world problem as a computational problem
has, broadly speaking, two steps: (i) abstracting the problem into a
mathematical formulation which captures the crux of the real-world problem, and
(ii) asking for a best solution to the mathematical problem.

Consider the following scenario. Dr.~\Ocal organizes a panel discussion, and has
a shortlist of candidates to invite. From that shortlist, Dr.~\Ocal wants to
invite as many candidates as possible, such that each of them will bring an
individual contribution to the panel. Given two candidates \(A\) and \(B\), it
may not be beneficial to invite both \(A\) and \(B\), for various reasons: their
areas of expertise or opinions may be too similar for both to make a
distinguishable contribution, or it may be preferable not to invite more than
one person from each institution. It may even be the case that \(A\) and \(B\)
do not see eye-to-eye on some issues which could come up at the discussion, and
Dr.~\Ocal wishes to avoid a confrontation.

A natural mathematical model to resolve Dr.~\Ocal's dilemma is as an instance of
the \VC problem: each candidate on the shortlist corresponds to a vertex, and for
each pair of candidates \(A\) and \(B\), we add the edge between \(A\) and \(B\)
if it is \emph{not} beneficial to invite both of them. Removing a smallest
\emph{vertex cover} in the resulting graph results in a largest possible set of
candidates such that each of them may be expected to individually contribute to
the appeal of the event.

Formally, a \emph{vertex cover} of an
undirected graph \(G\) is any subset \(S \subseteq V(G)\) of the vertex set of
\(G\) such that every edge in \(G\) has at least one end-point in \(G\). The \VC
problem asks for a vertex cover of the smallest size:

\optprobA{\VC}{Graph \(G\).}{A vertex cover \(S\) of \(G\) of the smallest size.}

While the above model does provide Dr.~\Ocal with a set of candidates to invite
that is \emph{valid} in the sense that each invited candidate can be expected to
make a unique contribution to the panel, a vast amount of \emph{side
  information} about the candidates is lost in the modeling process. This side
information could have helped Dr.~\Ocal to get more out of the panel discussion.
For instance, Dr.~\Ocal may have preferred to invite more well-known or
established people over `newcomers', if they wanted the panel to be highly
visible and prestigious; or they may have preferred to have more `newcomers' in
the panel, if they wanted the panel to have more outreach. Other preferences
that Dr.~\Ocal may have had include: to have people from many different cultural
backgrounds, to have equal representation of genders, or preferential
representation for affirmative action; to have a variety in the levels of
seniority among the attendants, possibly skewed in one way or the other. Other
factors, such as the total carbon footprint caused by the participants' travels,
may also be of interest to Dr.~\Ocal. This list could go on and on.

Now, it is possible to plug in some of these factors into the mathematical
model, for instance by including weights or labels. Thus a vertex weight could
indicate `how well-established' a candidate is. 
However, the complexity of the model grows fast with each additional
criterion. The classic field of multicriteria optimization~\cite{MC05} addresses the issue
of bundling multiple factors into the objective function, but it is seldom
possible to arrive at a balance in the various criteria in a way which captures
more than a small fraction of all the relevant side
information.
Moreover, several
side criteria may be conflicting or incomparable (or both); consider in
Dr.~\Ocal's case `maximizing the number of different cultural backgrounds' vs.\
`minimizing total carbon footprint.'

While Dr.~\Ocal's story is admittedly a made-up one, the \VC problem is in fact
used to model \emph{conflict resolution} in far more realistic settings. In each
case there is a \emph{conflict graph} \(G\) whose vertices correspond to
entities between which one wishes to avoid a conflict of some kind. There is an
edge between two vertices in \(G\) if and only if they could be in conflict, and
finding and deleting a smallest vertex cover of \(G\) yields a largest
conflict-free subset of entities. We describe three examples to illustrate the
versatility of this model. In each case it is intuitively clear, just like in
Dr.~\Ocal's problem, that formulating the problem as \VC results in a lot of
significant side information being thrown away, and that while finding a
smallest vertex cover in the conflict graph will give a \emph{valid} solution,
it may not really help in finding a \emph{best} solution, \emph{or even a
  reasonably good solution}. We list some side information that is lost in the
modeling process; the reader should find it easy to come up with any amount of
other side information that would be of interest, in each case.
\begin{description}
\item[Air traffic control.] Conflict graphs are used in the design of decision
  support tools for aiding Air Traffic Controllers (ATCs) in preventing untoward
  incidents involving aircraft~\cite{vela2011understanding,idan2010efficient}. 
  Each node in the graph \(G\) in this instance is an
  aircraft, and there is an edge between two nodes if the corresponding aircraft
  are at risk of interfering with each other. A vertex cover of \(G\)
  corresponds to a set of aircraft which can be issued \emph{resolution
    commands} which ask them to change course, such that afterwards there is no
  risk of interference.

  In a situation involving a large number of aircraft it is unlikely that
  \emph{every} choice of ten aircraft to redirect is \emph{equally} desirable.
  For instance, in general it is likely that (i) it is better to ask smaller
  aircraft to change course in preference to larger craft, and (ii) it is better
  to ask aircraft which are cruising to change course, in preference to those
  which are taking off or landing.  
\item[Wireless spectrum allocation.] Conflict graphs are a standard tool in
  figuring out how to distribute wireless frequency spectrum among a large set
  of wireless devices so that no two devices whose usage could potentially
  interfere with each other are allotted the same
  frequencies~\cite{gandhi2007general,hoefer2014approximation}. Each node in
  \(G\) is a user, and there is an edge between two nodes if (i)~the users
  request the same frequency, and (ii)~their usage of the same frequency has
  the potential to cause interference. A vertex cover of \(G\) corresponds to a
  set of users whose requests can be denied, such that afterwards there is no
  risk of interference.

  When there is large collection of devices vying for spectrum it is unlikely
  that \emph{every} choice of ten devices to deny the spectrum is \emph{equally}
  desirable. For instance, it is likely that denying the spectrum to a
  remote-controlled toy car on the ground is preferable to denying the spectrum
  to a drone in flight.
\item[Managing inconsistencies in database integration.] A database constructed
  by integrating data from different data sources may end up being inconsistent
  (that is, violating specified integrity constraints) even if the constituent
  databases are individually consistent. Handling these inconsistencies is a
  major challenge in database integration, and conflict graphs are central to
  various approaches for restoring consistency~\cite{chomicki2005minimal,arenas2003scalar,pema2011tractability,ioannou2013management}. 
  Each node in \(G\) is a database item, and there
  is an edge between two nodes if the two items together form an inconsistency.
  A vertex cover of \(G\) corresponds to a set of database items in whose
  \emph{absence} the database achieves consistency.

  In a database of large size it is unlikely that all data are created equal;
  some database items are likely to be of better relevance or usefulness than
  others, and so it is unlikely that \emph{every} choice of ten items to delete
  is \emph{equally} desirable.
\end{description}

Getting back to our first example, it seems difficult to help Dr.~\Ocal with
their decision by employing the `traditional' way of modeling computational
problems, where one looks for one best solution. If on the other hand, Dr.~\Ocal
was presented with a \emph{small set of good solutions} that in some sense are
\emph{far apart}, then they might hand-pick the list of candidates that they
consider the best choice for the panel and make a more informed decision.
Moreover, several forms of side-information may \emph{only become apparent once
  Dr.~\Ocal is presented some concrete alternatives}, and are more likely to be
retrieved from alternatives that look very different. That is, a bunch of good
quality, dissimilar solutions may end up capturing a lot of the ``lost'' side
information. And this applies to each of the other three examples as well. In
each case, finding one best solution could be of little utility in solving the
original problem, whereas finding a \emph{small set of solutions, each of good
  quality, which are not too similar to one another} may offer much more help.

To summarize, real-world problems typically have complicated side constraints,
and the optimality criterion may not be clear. Therefore, the abstraction to a
mathematical formulation is almost always a simplification, 
omitting important side information. 
There are at least two obstacles to
simply adapting the model by incorporating these secondary criteria into the
objective function or taking into account the side constraints: (i) they make
the model complicated and unmanagable, and (ii) more importantly, these criteria
and constraints are often not precisely formulated, potentially even unknown a priori. 
There may even be no sharp
distinction between optimality criteria and constraints (the so-called ``soft
constraints''). 

One way of dealing with this issue is to present a small number $r$ of
\emph{good} solutions and let the \emph{user }choose between them, based on all
the experience and additional information that the user has and that is ignored
in the mathematical model. Such an approach is useful even when the objective
can be formulated precisely, but is difficult to optimize: After generating $r$
solutions, each of which is \emph{good enough} according to some quality
criterion, they can be compared and screened in a second phase, evaluating their
exact objective function or checking additional side constraints. In this
context, it makes little sense to generate solutions that are very similar to
each other and differ only in a few features. It is desirable to present a
\emph{diverse} variety of solutions.

It should be clear that the issue is scarcely specific to \VC. Essentially
\emph{any} computational problem motivated by practical applications likely has
the same issue: the modeling process throws out so much relevant side
information that any algorithm which finds just one optimal solution to an input
instance may not be of much use in solving the original problem in practice. One
scenario where the traditional approach to modeling computational problems
fails completely is when computational problems may combined with a human sense of
aesthetics or intuition to solve a task, or even to stimulate inspiration. Some
early relevant work is on the problem of designing a tool which helps an
architect in creating a floor plan which satisfies a specified set of
constraints. In general, the number of feasible floor plans---those which
satisfy constraints imposed by the plot on which the building has to be erected,
various regulations which the building should adhere to, and so on---would be
too many for the architect to look at each of them one by one. 
Further, many of these plans
would be very similar to one another, so that it would be pointless for the
architect to look at more than one of these for inspiration.  
As an alternative to optimization for such problems, Galle proposed a
``Branch \& Sample'' algorithm for generating a ``limited, representative sample
of solutions, uniformly scattered over the entire solution
space''~\cite{galle1989branch}.

\paragraph{The Diverse \(X\) Paradigm.}
Mike Fellows has proposed \emph{the Diverse \(X\) Paradigm} as a solution for
these issues and others~\cite{fellows2018diverseXparadigm}. In this paradigm
``\(X\)'' is a placeholder for an optimization problem, and we study the
complexity---specifically, the fixed-parameter tractability---of the problem of
finding a few different good quality solutions for \(X\). Contrast this with the
traditional approach of looking for just one good quality solution. Let \(X\)
denote an optimization problem where one looks for a minimum-size subset of some
set; \VC is an example of such a problem. The generic form of \(X\) is then:

\optprobA{\(X\)}{An instance \(I\) of \(X\).}{A solution \(S\) of \(I\) of
  the smallest size.}

Here the form that a ``solution \(S\) of \(I\)'' takes is dictated by the
problem \(X\); compare this with the earlier definition of \VC.

The \emph{diverse} variant of problem \(X\), as proposed by Fellows, has the
form

\parprobA{\textsc{Diverse \(X\)}}{An instance \(I\) of \(X\), and positive
  integers \(k, r, t\).}{\((k, r)\)}{A set \(\Ss\) of \(r\) solutions of \(I\),
  each of size at most \(k\), such that a \emph{diversity measure} of \Ss is at
  least \(t\).}

Note that one can construct diverse variants of other kinds of problems as well,
following this model: it doesn't have to be a minimization problem, nor does the
solution have to be a subset of some kind. Indeed, the example about floor plans
described above has neither of these properties. What is relevant is that one
should have (i) some notion of ``good quality'' solutions (for \(X\), this
equates to a small size) and (ii) some notion of a set of solutions being
``diverse''.

\paragraph{Diversity measures.}
The concept of diversity appears also in other fields, and
there are many different ways to measure the diversity of 
a 
{collection}. 
For example, in
ecology, the diversity of a set of species (``biodiversity'') is a
topic that has become increasingly important in recent times,
see for example Solow and Polasky~\cite{Solow1994}.

Another possible viewpoint, in the context of multicriteria
optimization,
is to require that the sample of solutions should try
to represent the \emph{whole solution space}. This concept can be quantified for example by the geometric \emph{volume} of the represented
space~\cite{MaximumVolume,Hypervolume2d}, or by
the \emph{discrepancy}~\cite{Neumann:2018}.
See
\cite[Section~3]{defining-diversity-measures-2010} for an overview of
diversity measures in multicriteria
optimization. 

In this paper, we follow
the simple possibility of looking for a collection of good solutions
 that
have large \emph{distances} from each other,
in a sense that will be made precise below~{\eqref{eq:objectivemin}--\eqref{eq:objectivetot}}.
Direction~\eqref{eq:objectivetot}, i.e., taking the pairwise sum of all Hamming distances, has been taken by many practical papers in the area of genetic algorithms, see e.g.~\cite{GaborBPS18,MorrisonJ01}.
This now classical approach can be traced as far back as 1992~\cite{LouisR92}.
In~\cite{WinebergO03a}, it has been boldly stated that this measure (and its variations) is one of the most broadly used measures in describing population diversity within genetic algorithms.
One of its advantages is that it can be computed very easily and efficiently unlike many other measures, e.g., some geometry or discrepancy based measures.

\subsection{Our problems and results.}

In this work we focus on diverse versions of two minimization problems, \dHS and
\FVS, whose solutions are subsets of a finite set. \dHS is in fact a
\emph{class} of such problems which includes \VC, as we describe below. 
We will consider two natural diversity measures
for these problems: the minimum Hamming distance between any two solutions, and
the sum of pairwise Hamming distances of all the solutions.

The \emph{Hamming distance} between two sets \(S\) and \(S'\), or \emph{the size
  of their symmetric difference}, is
\begin{displaymath}
  d_H(S,S') := | (S \setminus S') \cup (S'\setminus S) |.
\end{displaymath}
We use
\begin{equation}
  \label{eq:objectivemin}
  \mindiversity(S_1,\ldots, S_r) := \min_{1 \leq i<j\leq r}\Hammingdistance(S_i,S_j)
\end{equation}
to denote the minimum Hamming distance between any pair
of sets in a collection of finite sets,
and
\begin{equation}
  \label{eq:objectivetot}
  \sumdiversity(S_1,\ldots, S_r) := \sum_{1\le i< j\le r}  d_H(S_i,S_j)
\end{equation}
to denote the sum of all pairwise
Hamming distances.
(In Section~\ref{modeling}, we will discuss some issues
with the latter formulation.)

A \emph{feedback vertex set} of a graph \(G\) is any subset \(S \subseteq V(G)\)
of the vertex set of \(G\) such that the graph \(G-S\) obtained by deleting the
vertices in \(S\) is a \emph{forest}; that is, contains no cycle. 

\optprobA
	{\FVS}
	{A graph $G$.}
	{A feedback vertex set of $G$ of the smallest size.}

More
generally, a \emph{hitting set} of a collection \F of subsets of a universe
\(U\) is any subset \(S \subseteq U\) such that every set in the family \F has a
non-empty intersection with \(S\). For a fixed positive integer \(d\) the \dHS
problem asks for a hitting set of the smallest size of a family \F of
\(d\)-sized subsets of a finite universe \(U\):

\optprobA{\dHS}{A finite universe \(U\) and a family \F of subsets of \(U\),
  each of size at most \(d\).}{A hitting set \(S\) of \F of the smallest size.}

Observe that both \VC and \FVS are special cases of finding a smallest hitting
set for a family of subsets. \VC is also an instance of \dHS, with \(d = 2\):
the universe \(U\) is the set of vertices of the input graph and the family \F
consists of all sets \(\{v,w\}\) where \(vw\) is an edge in \(G\). There is no
obvious way to model \FVS as a \dHS instance, however, because the cycles in the
input graph are not necessarily of the same size.

In this work, we consider the following problems in the \textsc{Diverse $X$} paradigm.
Using $\sumdiversity$ as the diversity measure, we consider \rDkdHS and \rDkFVS,
where $X$ is \dHS and \FVS, respectively.
Using $\mindiversity$ as the diversity measure, we consider \rMDkdHS and \rMDkFVS,
where $X$ is \dHS and \FVS, respectively.

In each case we show that the problem is fixed-parameter tractable (\FPT), 
with the following running times:

\begin{Theorem}\label{thm:rDkdHS_is_FPT}
  \rDkdHS can be solved in time $r^{2}d^{kr}\cdot |U|^{O(1)}$.
\end{Theorem}

\begin{Theorem}\label{thm:fvs}
 \rDkFVS can be solved in time $2^{7kr}\cdot n^{O(1)}$.
\end{Theorem}

\begin{Theorem}\label{thm:rMDkdHS_is_FPT}
  \rMDkdHS can be solved in time
  \begin{itemize}
  \item $2^{kr^2}\cdot (kr)^{O(1)}$
    if $|U| < kr$ and
  \item $d^{kr}\cdot |U|^{O(1)}$ otherwise.
  \end{itemize}
\end{Theorem}

\begin{Theorem}\label{thm:fvsminmax}
  \rMDkFVS can be solved in time $2^{kr\cdot \max (r,7+\log_2(kr))}\cdot (nr)^{O(1)}$.
\end{Theorem}

Defining the diverse versions \rDkVC and \rMDkVC of \VC in a similar manner as
above, we get

\begin{Corollary}\label{cor:diverseVC_are_FPT}
  \rDkVC can be solved in time $2^{kr}\cdot n^{O(1)}$. \rMDkVC can be solved in
  time
  \begin{itemize}
  \item $2^{kr^2}\cdot (kr)^{O(1)}$ if $n < kr$ and
  \item $2^{kr}\cdot n^{O(1)}$ otherwise.
  \end{itemize}
\end{Corollary}

\paragraph{Related Work.}
The parameterized complexity of finding a diverse collection of good-quality
solutions to algorithmic problems seems to be largely unexplored. To the best of
our knowledge, the only existing work in this area consists of: (i) a privately
circulated manuscript by Fellows~\cite{fellows2018diverseXparadigm} which
introduces the Diverse \(X\) Paradigm and makes a forceful case for its
relevance, and (ii) a manuscript by Baste et al.~\cite{arBaFeJaOlRo2019}
which applies the Diverse \(X\) Paradigm to
\emph{vertex-problems} with the \emph{treewidth} of the input graph as an extra
parameter. In this context a \emph{vertex-problem} is any problem in which the
input contains a graph \(G\) and the solution is some subset of the vertex set
of \(G\) which satisfies some problem-specific properties. Both \VC and \FVS are
vertex-problems in this sense, as are many other graph problems. The
\emph{treewidth} of a graph is, informally put, a measure of how tree-like the
graph is. See, e.g., \cite[Chapter~7]{CyFoKoLoMaPiPiSa2015} for an introduction
of the use of the treewidth of a graph as a parameter in designing \FPT
algorithms. The work by Baste et al.~\cite{arBaFeJaOlRo2019} shows how to
convert essentially any treewidth-based dynamic programming algorithm for
solving a vertex-problem, into an algorithm for computing a diverse set of \(r\)
solutions for the problem, with the diversity measure being the sum
$\sumdiversity$ of Hamming distances of the solutions. This latter algorithm is
\FPT in the combined parameter \((r, w)\) where \(w\) is the treewidth of the
input graph. As a special case, they obtain a running time of
$\Ocal((2^{k+2}(k+1))^r kr^2n)$ for \rDkVC.
Further, they show that
the \(r\)-\textsc{Diverse} versions (i.e., where the diversity measure is
$\sumdiversity$) of a handful of problems have polynomial kernels. In
particular, they show that \rDkVC has a kernel with \(\Ocal(k(k + r))\)
vertices, and that \rDkdHS has a kernel with a universe size of
\(\Ocal(k^{d} + kr)\).

\paragraph{Organization of the rest of the paper.}
In \autoref{sec:prelims} we list some definitions which we use in the rest of
the paper. In \autoref{sec:framework} we describe a generic framework which can
be used for computing solution families of maximum diversity for a variety of
problems whose solutions form subsets of some finite set. We prove
\autoref{thm:rDkdHS_is_FPT} in Section~\ref{sec:diversedHS} and
\autoref{thm:fvs} in \autoref{Diverse_FVS}. In \autoref{modeling}
we discuss some potential pitfalls in using \(\sumdiversity\) as a measure of
diversity. In \autoref{max} we prove \autoref{thm:rMDkdHS_is_FPT} and
\autoref{thm:fvsminmax}. We conclude in \autoref{conclusion}.

\section{Preliminaries}\label{sec:prelims}

Given two integers $p$ and $q$, we denote by $[p,q]$ the set of all integers $r$
such that $p \leq r \leq q$ holds. Given a graph $G$, we denote by $V(G)$ (resp.
$E(G)$) the set of \emph{vertices} (resp. \emph{edges}) of $G$. For a subset
$S \subset V(G)$ we use $G[S]$ to denote the subgraph of $G$ induced by $S$, and
$G \setminus S$ for the graph $G[V(G) \setminus S]$. A set $S \subseteq V(G)$ is
a vertex cover (resp. a feedback vertex set) if $G\setminus S$ has no edge
(resp. no cycle). Given a graph $G$ and a vertex $v$ such that $v$ has exactly
two neighbors, say $w$ and $w'$, \emph{contracting} $v$ consists in removing the
edges $\{v,w\}$ and $\{v,w'\}$, removing $v$ and adding the edge $\{w,w'\}$.
Given a graph $G$ and a vertex $v\in V(G)$, we denote by $\delta_G(v)$ the
\emph{degree} of $v$ in $G$. For two vertices \(u,v\) in a connected graph \(G\)
we use \(\texttt{dist}_T(u,v)\) to denote the \emph{distance} between \(u\) and
\(v\) in \(G\), which is the length of a shortest path in \(G\) between \(u\)
and \(v\).

A \emph{deepest leaf} in a tree $T$ is a vertex $v \in V(T)$ such that there
exists a root $r \in V(T)$ satisfying
$\texttt{dist}_T(r,v) = \max_{u \in V(T)} \texttt{dist}_T(r,u)$. A \emph{deepest
  leaf} in a forest $F$ is a deepest leaf in some connected component of $F$.
A deepest leaf $v$ has the property that there is another leaf in the
tree at distance at most 2 from $v$ unless $v$ is an isolated vertex
or $v$'s neighbor has degree~2.

The objective function $\sumdiversity$ in
\eqref{eq:objectivetot} has an alternative representation in
terms of frequencies of occurrence~\cite{arBaFeJaOlRo2019}:
If $y_v$ is the number of sets of $\{S_1, \ldots, S_r\}$ in which $v$ appears,
then
\begin{equation}
  \label{eq:obj-alternative}
  \sumdiversity(S_1,\ldots, S_r) = \sum_{v \in U} y_v(r-y_v).
\end{equation}

\paragraph{Auxiliary problems.}
We define two auxiliary problems that we will use in some of the algorithms 
presented in \autoref{sec:framework}.
In the \maxcostflow problem, we are given a directed graph $G$, a \emph{target} $d \in \bR^+$,
a \emph{source vertex} $s \in V(G)$, a \emph{sink vertex} $t \in V(G)$,
and for each edge $(u, v) \in E(G)$, 
a \emph{capacity} $c(u, v) > 0$, and a \emph{cost} $a(u, v)$.
A \emph{$(s, t)$-flow}, or simply \emph{flow} in $G$ is a function $f \colon E(G) \to \bR$, such that 
for each $(u, v) \in E(G)$, $f(u, v) \le c(u, v)$, and for each vertex
$v \in V(G) \setminus \{s, t\}$, $\sum_{(u, v) \in E(G)} f(u, v) = \sum_{(v, u) \in E(G)} f(v, u)$.
The \emph{value} of the flow $f$ is $\sum_{(s, u) \in E(G)} f(s, u)$ and the 
\emph{cost} of $f$ is $\sum_{(u, v) \in E(G)} f(u, v) \cdot a(u, v)$.
The objective of the \maxcostflow problem is to find the maximum cost $(s, t)$-flow of value $d$.

The second problem is the \MWBM problem. Here, we are given an undirected edge-weighted graph $G$, 
and for each vertex $v \in V(G)$, a \emph{supply} $b(v)$. The goal is to find a set of edges 
$M \subseteq E(G)$ of maximum total weight such that each vertex $v \in V(G)$ is incident with at most $b(v)$ edges in $M$.

\section{A Framework for Maximally Diverse Solutions}
\label{sec:framework}

In this section we describe a framework for computing solution families of
maximum diversity for a variety of hitting set problems. This framework requires
that 
the solutions form a family of subsets of a ground set $U$ which is upward
closed:
Any superset $T\supseteq S$ of a solution $S$ is also a solution.

The approach is as follows:
In a first phase, we enumerate the class $\mathcal{S}$ of all 
\emph{minimal solutions} of size at most~$k$. (A larger class $\mathcal{S}$ is also
fine as long as it is guaranteed to contain all minimal solutions of
size at most $k$.)
Then we form all $r$-tuples $(S_1,\ldots,S_r) \in
\mathcal{S}^k$. For each such family $(S_1,\ldots,S_r)$
, we try to \emph{augment}
it to a family
 $(T_1,\ldots,T_r)$ under the constraints $T_i\supseteq S_i$ and
 $|T_i|\le k$, for each $i \in [1,r]$, in such a way that
$\sumdiversity(T_1,\ldots,T_r)$ is maximized.

For this augmentation problem, we propose a network flow model that computes an
optimal augmentation in polynomial time, see Section~\ref{sec:augmentation}.
This has to be repeated for each family, $O(|\mathcal{S}|^r)$ times. The first
step, the generation of $\mathcal{S}$, is problem-specific.
Section~\ref{sec:diversedHS} shows how to solve it for \dHS. In
Section~\ref{Diverse_FVS}, 
we will adapt our approach to deal with \FVS.

\subsection{Optimal Augmentation}
\label{sec:augmentation}

Given a universe $U$ and a set $\mathcal{S}$ of subsets of $U$, the problem
$\texttt{diverse}_{r,k}(\mathcal{S})$ consists in finding an
$r$-tuple $(S_1,\ldots, S_r)$ that maximizes
$\sumdiversity(S_1, \ldots, S_r)$, over all $r$-tuples
$(S_1, \ldots, S_r)$ such that for each $i \in [1,r]$, $|S_i| \le k$ and
there exists $S \in \mathcal{S}$ such that
$S \subseteq S_i \subseteq U$.

\begin{Theorem}
  \label{th:optaug}
  Let $U$ be a finite universe, $r$ and $k$ be two integers and  $\mathcal{S}$ be a set of $s$ subsets of $U$.
  $\textup{\texttt{diverse}}_{r,k}(\mathcal{S})$ can be solved in time $r^2s^{r} \cdot |U|^{O(1)}$.
\end{Theorem}

\begin{proof}
The algorithm that proves Theorem~\ref{th:optaug} starts by
enumerating all
\(r\)-tuples $(S_1,S_2,\ldots,S_r) \in \mathcal{S}^r$ of elements from
$\mathcal{S}$.
For each of these \(s^{r}\) \(r\)-tuples we try to augment each $S_i$, using
elements of $U$, in such a way that the diversity $d$ of the resulting tuple
$(T_1,\ldots,T_r)$ is maximized and such that for each $i \in [1,r]$,
$S_i \subseteq T_i \subseteq U$ and $|T_i|\le k$. It is clear that this
algorithm will find the solution to $\texttt{diverse}_{r,k}(\mathcal{S})$.

We show how to model this problem as a maximum-cost network flow problem with
piecewise linear concave costs.
This problem can be solved in polynomial time.
(See for example \cite{Tarjan} for basic notions about network flows.)

Without loss of generality, let \(U = \{1, 2, \dotsc, n\}\). We use a variable
$0\le x_{ij}\le1$ to decide whether element $j$ of \(U\) should belong to set
$T_i$. In an optimal flow, these values are integral. Some of these variables
are already fixed because $T_i$ must contain~$S_i$:
\begin{equation}
  \label{eq:given}
  x_{ij}=1 \text{ for }j\in S_i
\end{equation}
The size of $T_i$ must not exceed $k$:
\begin{equation}
  \label{eq:bound}
  \sum_{j=1}^n x_{ij}\le k, \text{ for } i=1,\ldots,r
\end{equation}
Finally, we can express the number $y_j$ of sets $T_i$ in which an element $j$
occurs:
\begin{equation}
  \label{eq:multiplicity}
 y_j =  \sum_{i=1}^r x_{ij}, \text{ for } j=1,\ldots,n
\end{equation}
These variables $y_j$ are
the variables in terms of which the objective
function
 \eqref{eq:obj-alternative}
 is expressed:
\begin{equation}
  \label{eq:obj}
 \text{maximize } \sum_{j=1}^n y_j(r-y_j)
\end{equation}
These constraints can be modeled by a network as shown in
Figure~\ref{fig:network}. There are nodes $T_i$ representing the sets $T_i$ and
a node $V_j$ for each element \(j \in U\). In addition, there is a source $s$
and a sink~$t$. The arcs emanating from $s$ have capacity $k$. Together with the
flow conservation equations at the nodes $T_i$, this models the
constraints~\eqref{eq:bound}. Flow conservation at the nodes $V_j$ gives rise to
the flow variables $y_j$ in the arcs leading to $t$ according
to~\eqref{eq:multiplicity}. The arcs with fixed flow \eqref{eq:given} could be
eliminated from the network, but for ease of notation, we leave them in the
model. The only arcs that carry a cost are the arcs leading to $t$, and the
costs are given by the concave function~\eqref{eq:obj}.

\begin{figure}
  \centering
  \includegraphics {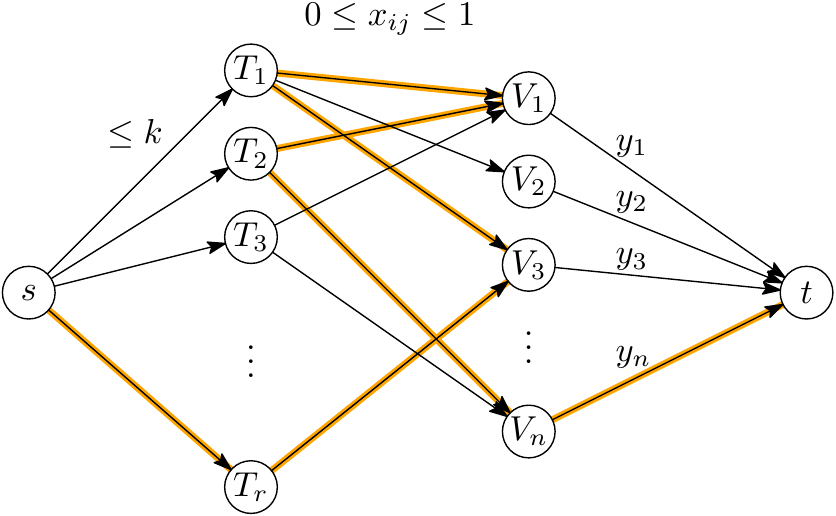}
  \caption{The network.
The middle layer between the vertices $T_i$ and $V_j$ is a complete
bipartite graph, but only a few selected arcs are shown.
    A potential augmenting path is highlighted.}
  \label{fig:network}
\end{figure}
There is now a one-to-one correspondence between integral flows from
$s$ to $t$ in the
network and solutions $(T_1,\ldots,T_r)$, and the cost of the flow is
equal to the diversity
 \eqref{eq:objectivetot} or
 \eqref{eq:obj-alternative}.
We are thus looking for a flow of maximum cost. 
The \emph{value} of the flow (to total flow out of $s$)
can be arbitrary. (It is equal to the sum of the sizes of the sets $T_i$.)

 The concave arc costs~\eqref{eq:obj} on
 the arcs leading to $t$ can be modeled in
a standard way by
multiple arcs.
Denote the concave cost function by $f_y:=y(r-y)$, for $y=0,1,\ldots,r$.
Then each arc $(V_i,t)$ in the last layer is replaced by $r$ parallel
arcs of capacity~1 with
costs $f_1-f_0$,  $f_2-f_1$, \dots,
$f_r-f_{r-1}$.
This sequence of values $f_y-f_{y-1}=r-2y+1$
is decreasing, starting out with positive values and
ending with negative values.
If the total flow along such a bundle is $y$, the
 maximum-cost way to distribute this flow is to fill the first $y$
 arcs to capacity, for a total cost of $(f_1-f_0) + (f_2-f_1)+ \dots+
 (f_y-f_{y-1}) =  f_y-f_0 =f_y$, as desired.

An easy way to compute a maximum-cost flow is the longest augmenting path
method.  (Commonly it is presented as the \emph{shortest} augmenting
path method for the \emph{minimum}-cost flow.)
This holds for the classical flow model where the cost on each arc is a linear
function of the flow.
An augmenting path is
a path in the residual network with respect to the current flow, and
the cost coefficient of an arc in such a path must be taken with opposite sign if
it is traversed in the direction opposite to the original graph.

\begin{Proposition}[The shortest augmenting path algorithm, cf.\ {\cite[Theorem 8.12]{Tarjan}}]
  Suppose a maximum-cost flow among all flows of value $v$ from $s$ to $t$ is given.
  Let $P$ be a maximum-cost augmenting path from $s$ to $t$.
  If we augment the
  flow along this path, this results in a
  new flow, of some value $v'$.
 Then the new flow is a maximum-cost flow  among all flows of value $v'$ from $s$ to $t$.  
\end{Proposition}

Let us apply this algorithm to our network.
We initialize the constrained flow variables $x_{ij}$ according to
\eqref{eq:given} to 1 and all other variables
 $x_{ij}$ to 0. This corresponds to the original solution
 $(S_1,S_2,\ldots,S_r)$, and it is clearly the optimal flow of
 value
 $\sum_{i=1}^r|S_i|$ because it is the only feasible flow of this value.

 We can now start to find augmenting paths.
Our graph is bipartite, and augmenting paths have a very simple
structure:
They start in $s$, alternate back and forth between the $T$-nodes and
the $V$-nodes, and finally make a step to~$t$.
Moreover, in our network, all costs are zero except in the last layer,
and an augmenting path contains precisely one arc from this layer.
Therefore,
\emph{
the cost of an
augmenting path is simply the cost of the final arc}.

The flow variables in the final layer are never decreased.
The resulting algorithm has therefore a simple greedy-like structure.
Starting from the initial flow,
we first try to saturate as many of the arcs of cost $f_1-f_0$ as
possible.
Next,
we try to saturate as many of the arcs of cost $f_2-f_1$ as
possible, and so on.
Once the incremental cost $f_{y+1}-f_y$ becomes negative, we stop.

Trying to find an augmenting path whose last arc is one of the arcs
of cost $f_{y+1}-f_y$, for fixed $y$, is a reachability problem in the residual graph,
and
it can be solved by graph search in $O(nr)$ time because the network has $O(nr)$
vertices.
Every augmentation increases the flow value by 1 unit. Thus, there are
at most $kr$ augmentations, for a total runtime of $O(kr^2n)$.
\end{proof}

\subsection{Faster Augmentation}
We can obtain faster algorithms by using more advanced
network algorithms from the literature.
We will derive one such algorithm here. The best choice depends
on the relation between $n$, $k$, and $r$.
We will apply the following result about \emph{$b$-matchings}, which
 are generalizations of matchings: Each node $v$
has a given \emph{supply} $b(v)$, specifying that $v$ should be incident to at
most $v$ edges.
\begin{Proposition}
  [\cite{ahuja94}]
  \label{b-matching}
  A maximum-weight
  $b$-matching in a bipartite graph with $N_1+N_2$ nodes
on the two sides of the bipartition
  and $M$
edges that have integer weights between $0$ and $W$ can be found in time
$O(N_1M \log(2+\frac{{N_1\!}^2}M \log(N_1W)))$.
\end{Proposition}

We will describe below how the network flow problem from above can be
converted into a $b$-matching problem with $N_1=r+1$ plus $N_2=n$ nodes
and $M=2rn$ edges of weight at most $W=2r$.
Plugging these values into Proposition~\ref{b-matching} gives
a running time of
$O(r^2n\log(2+\frac rn\log(r^2)))
=O(r^2n\max\{1,\log\frac{r\log r}n\})$
for finding an optimal augmentation.
This improves over the run time $O(r^2nk)$ from the previous section
unless $r$ is extremely large (at least $2^k$).

From the network of
Figure~\ref{fig:network}, we keep the two layers of nodes $T_i$ and
$V_j$.
Each vertex $T_i$ gets a supply of $b(T_i):= k$, and
each vertex $V_j$ gets a supply of $b(V_j):= r$.
To mimic the piecewise linear costs on the arcs $(V_j,t)$ in the
original network, we introduce
$r$ parallel \emph{slack edges} from a new source vertex $s'$ to each vertex $V_i$.
The costs are as follows. Let $g_1>g_2>\cdots>g_r$ with
$g_y=f_y-f_{y-1}$ denote the costs in the last layer of the original
network,
and let $\hat g:=r$. Since $g_1=r-1$, this is larger than all costs.
Then every edge $(T_i,V_j)$ from the original network gets a weight of $\hat g$,
and the $r$ new slack edges entering each $V_j$ get positive weights $\hat g-g_1,
\hat g-g_2, \ldots,
\hat g-g_r$. We set the supply of the extra source node to $b(s') := rn$,
which imposes no constraint on the number of incident edges.

Now suppose that we have a solution for the original network
 in which the total flow into vertex $V_j$ is $y$.
 In the corresponding $b$-matching, we can then use $b(V_j)-y$ = $r-y$ of the slack edges incident
 to $V_j$. The $r-y$ maximum-weight slack edges have weights
 $\hat g-g_r,
 \hat g-g_{r-1},
 \ldots
 \hat g-g_{y+1}$.
 The total weight of the edges incident to $V_j$ is therefore
 \begin{displaymath}
   r\hat g-g_r-g_{r-1}-\cdots- g_{y+1}
   =   r\hat g+(g_1+g_2+\cdots+ g_{y}),
 \end{displaymath}
 using the equation
 $g_1+g_2+\cdots+ g_{r}= f_r-f_0=0$.
Thus, up to an addition of the constant $nr\hat g$, the maximum weight
of a
$b$-matching agrees with the maximum cost of a flow in the original
network.

\subsection{Diverse Hitting Set}\label{sec:diversedHS}

In this section we show how to use the optimal augmentation technique developed
in Section~\ref{sec:framework} to solve \rDkdHS.
For this we use the following folklore lemma about {minimal} hitting
sets. 
 \begin{Lemma}
   \label{lemma:minimaldHS}
   Let \((U, \F)\) be an instance of \dHS, and let $k$ be an integer. There are
   at most $d^k$ inclusion-minimal hitting sets of \F of size at most $k$, and
   they can all be enumerated in time $d^k |U|^2 $.
\end{Lemma}

Combining Lemma~\ref{lemma:minimaldHS} and Theorem~\ref{th:optaug}, we obtain the following result.

\begingroup
\def\thetheorem{\ref{thm:rDkdHS_is_FPT}}
\begin{Theorem}
  \rDkdHS can be solved in time $r^{2}d^{kr}\cdot |U|^{O(1)}$.
\end{Theorem}
\addtocounter{theorem}{-1}
\endgroup
\begin{proof}
  Using Lemma~\ref{lemma:minimaldHS}, we can construct the set $\mathcal{S}$ of
  all inclusion-minimal hitting sets of $\F$, each of size at most $k$. Note
  that the size of $\mathcal{S}$ is bounded by $d^k$. As every superset of an
  element of $\mathcal{S}$ is also a hitting set, the theorem follows directly
  from Theorem~\ref{th:optaug}.
\end{proof}

\section{Diverse Feedback Vertex Set}
\label{Diverse_FVS}

A \emph{feedback vertex set} (FVS) (also called a \emph{cycle
  cutset}) of a graph \(G\) is any subset \(S\subseteq{}V(G)\) of
vertices of \(G\) such that every cycle in \(G\) contains at least
one vertex from \(S\). The graph \(G-S\) obtained by deleting
\(S\) from \(G\) is thus an acyclic graph. Finding an FVS of small
size is an NP-hard problem~\cite{KarpNPc} with a number of applications in
Artificial Intelligence, many of which stem from the fact that many
hard problems become easy to solve in acyclic graphs.
An example for this is the Propositional Model Counting (or \#SAT) problem which asks for the number of
satisfying assignments for a given CNF formula, and has a number of
applications, for instance in
planning~\cite{domshlak2006fast,palacios2005pruning} and in probabilistic
inference problems such as Bayesian
reasoning~\cite{bacchus2003algorithms,littman2001stochastic,sang2005performing,apsel2012lifted}.
A popular approach to solving \#SAT consists of first finding a small FVS \(S\)
of the CNF formula. Assigning values to all the variables in \(S\) results in an
acyclic instance of CNF. The algorithm assigns all possible sets of values to
the variables in \(S\), computes the number of satisfying assignments of the
resulting acyclic instances, and returns the sum of these
counts~\cite{dechter2003constraint}.

In this section, we focus on the \textsc{Diverse Feedback Vertex Set} problem and prove the following theorem.

\begingroup
\def\thetheorem{\ref{thm:fvs}}
\begin{Theorem}
 \rDkFVS can be solved in time $2^{7kr}\cdot n^{O(1)}$.
\end{Theorem}
\addtocounter{theorem}{-1}
\endgroup

In order to solve \textsc{$r$-Diverse $k$-Feedback Vertex Set}, one natural way
would be to generate every feedback vertex set of size at most $k$ and then
check which set of $k$ solutions provide the required sum of Hamming distances.
Unfortunately, the number of feedback vertex set is not \FPT parameterized by
$k$. Indeed, one can consider a graph containing $k$ cycle of size
$\frac{n}{k}$, leading to $\left(\frac{n}{k}\right)^k$ different feedback vertex
sets of size $k$.

We avoid this problem by generating all such small feedback vertex sets up to some equivalence of degree two vertices.
We obtain an exact and efficient description of all feedback vertex sets of size at most $k$, which is formally captured by Lemma~\ref{lem:eqFVS}.
A \emph{class of solutions} of a graph $G$, is a pair $(S,\ell)$ such
that $S \subseteq V(G)$ and $\ell : S \to 2^{V(G)}$ is a function such
that for each $u \in S$, $u \in \ell(u)$, and for each $u,v \in S$, $u
\not = v$, $\ell(u) \cap \ell(v) = \emptyset$.
Given a class of solutions $(S,\ell)$, we define $\sol(S,\ell) = \{ S' \subseteq V(G): |S'| = |S| \text{ and } \forall v \in S, |S' \cap \ell(v)| = 1\}$.
A \emph{class of FVS solutions} is a class of solutions $(S,\ell)$ such that each $S' \in \sol(S,\ell)$ is a feedback vertex set of $G$.
Moreover, if $S' \in \sol(S,\ell)$ and $S' \subseteq S'' \subseteq V(G)$, we say that $S''$ is \emph{described} by $(S,\ell)$.
Note that $S''$ is also a feedback vertex set.
In a class of FVS solutions $(S,\ell)$, the meaning 
of the function $\ell$ is that,
for each cycle $C$ in $G$, there exists $v\in S$ such that each element of $\ell(v)$ hits $C$.
This allows us to group related solutions
 into only one set
$\sol(S,\ell)$.

\begin{Lemma}\label{lem:eqFVS}
  Let $G$ be a $n$-vertex graph.
  There exists a set $\mathcal{S}$ of classes of FVS solutions of $G$
  of size at most $2^{7k}$ such that each 
  feedback vertex set of size at most $k$ is described by an element of $\mathcal{S}$.
  Moreover, $\mathcal{S}$ can be constructed in time {$2^{7k}\cdot n^{O(1)}$}.
\end{Lemma}

\begin{proof}
  Let $G$ be a $n$-vertex graph.
  We start by generating a feedback vertex set $F\subseteq V$ of size
  at most~$k$.
  The current best deterministic
  algorithm  for this by Kociumaka and Pilipczuk~\cite{KOCIUMAKA2014556} 
  finds such a 
  set in time $3.62^k\cdot n^{{O}(1)}$.
In the following, we use the ideas used for the iterative compression approach~\cite{ReSmVe2004}.

For each subset $F'\subseteq F$,
we initiate a
branching process
by setting $A := F'$, $B := F-F'$,
and $G' := G$.
Observe that initially, as
$B 
\subseteq F$ and $|F|\leq k$,
the graph $G[B]$
has at most $k$ components.
In the branching process, we will add more vertices to $A$ and $B$,
and we will remove vertices and edges from $G'$,
but we will maintain the property that $A \subseteq V(G')$ and $B \subseteq V(G')$.
The set $C$ will always denote the vertex set $V(G') \setminus (A \cup B)$.
Note that $G'[C]$ is initially a forest; we ensure that it always remains a forest.

We also initialize a function $\ell\colon V(G) \to 2^{V(G)}$
by setting $\ell(v)=\{v\}$ for each $v \in V(G)$.
This function will keep information about vertices that are deleted
from $G$.
While searching for a feedback vertex set,
we consider only feedback vertex sets that contain all vertices of $A$ but no vertex of $B$.
Vertices in $C$ 
are still undecided.
The function $\ell$ will maintain the invariant that for each $v \in
V(G')$, $\ell(v) \cap V(G') = \{v\}$,  and
for each $v \in C$,  all vertices of $\ell(v)$ intersect exactly the same cycles in $G \setminus A$.
Moreover, for each $v \in A$, the value $\ell(v)$ is fixed and will not be modified anymore in the branching process.
During the branching process, we will progressively increase the size
of $A$, $B$, and the sets $\ell(v)$, $v \in V(G)$.

By \emph{reducing} $(G',A,B,\ell)$ we mean that we apply the 
following rules exhaustively.
\begin{itemize}
\item If there is a $v \in C$ such that $\delta_{G'[B\cup C]}(v) \leq 1$,  we delete $v$ from $G'$.
\item If there is an edge $\{u,v\} \in E(G'[C])$ such that
  $\delta_{G'[B \cup C]}(u) = \delta_{G'[B \cup C]}(v) = 2$,
  we contract $u$ in $G'$ and set $\ell(v)\df \ell(v)\cup \ell(u)$.
\end{itemize}
These are classical preprocessing rules for the \textsc{Feedback Vertex Set} problem, see for instance~\cite[Section 9.1]{CyFoKoLoMaPiPiSa2015}.
Indeed, vertices of degree one cannot appear in a cycle, and consecutive vertices of degree $2$ hit exactly the same cycles.
After this preprocessing, there are no adjacent degree-two vertices
and no degree-one vertices in $C$. (Degrees are measured in
$G'[B\cup C]$.)

We start to describe the branching procedure.
We work on the tuple $(G',A,B,\ell)$.
After each step, the value $|A|-\texttt{cc}(B)$ will increase,
where $\texttt{cc}(B)$ denotes the number of connected components of $G'[B]$.

At each step of the branching we do the following.
If $|A| > k$ or if $G'[B]$ contains a cycle, we immediately stop this branch as there is no solution to be found in it.
If $A$ is a feedback vertex set of size at most $k$, then $(A,\ell|_A)$ is a class of FVS solutions, we add it to $\mathcal{S}$ and stop working on this branch.
Otherwise, we reduce $(G',A,B,\ell)$. We pick a deepest leaf $v$ in $G'[C]$
and apply one of the two following cases, depending on the vertex $v$.

\begin{itemize}
\item {\bf Case 1:}
  The vertex $v$ has at least two neighbors in $B$ (in the graph~$G'$).
  
  If there is a path in $B$ between two neighbors of $v$, then we have to put $v$ in $A$, as otherwise this path together with $v$ will induce a cycle.
  If there is no such path, we branch on both possibilities,
  inserting $v$ either into $A$ or into $ B$.
  
\item {\bf Case 2:}
   The vertex $v$ has at most one neighbor in $B$.

  Since
  $v$ is a leaf in $G'[C]$, it has at most one neighbor also in $C$.
  On the other hand, we know that $v$ has degree at least 2 in
  $G'[B\cup C]$. Thus, $v$ has exactly one neighbor in $B$ and one
neighbor in $C$, for a degree of $2$ in $G'[B \cup C]$.
Let $p$ be the neighbor in $C$.
Again, as we have reduced $(G',A,B,\ell)$,  the degree of $p$ in
$G'[B\cup C]$ 
is at least $3$.
  So either it has a neighbor in $B$, or, as $v$ is a deepest leaf, it
  has another child, say $w$, that is also a leaf in $G'[C]$, and $w$
  has therefore a neighbor in $B$.
  We branch on the at most $2^3=8$ possibilities to allocate $v$, $p$,
  and $w$ if considered, between $A$ and $B$, taking care not to
  produce a cycle in~$B$.
\end{itemize}

In both cases, either we put at least one vertex in $A$, and so
$|A|$ increases by one, or all considered vertices are added to~$B$.
In the latter case,
the considered vertices are connected, at least two of them have a neighbor in $B$, and no
cycles were created; therefore, the number of components in $B$ drops by one.
Thus $|A|-\texttt{cc}(B)$
increases by at least one.
As $-k\le |A|-\texttt{cc}(B) \le k$,
there can be at most  $2k$ branching steps.

Since we branch at most $2k$ times and at each branch we have at most
$2^3$ possibilities, the branching tree has at most $2^{6k}$ leaves.
So, for each of the at most $2^k$
 subsets $F'$ of $F$,
 we add at most $2^{6k}$ elements to $\mathcal{S}$.

It is clear 
that we have obtained all solutions of FVS
and they are described by the classes of FVS solutions in
$\mathcal{S}$, which is of size $2^{7k}$.
\end{proof}

\begin{proof}[Proof of Theorem~\ref{thm:fvs}]
We  generate all $2^{7kr}$ $r$-tuples of the classes of solutions given by
Lemma~\ref{lem:eqFVS}, with repetition allowed.

We now consider each $r$-tuple $((S_1,\ell_1),(S_2,\ell_2),\ldots,(S_r,\ell_r)) \in \mathcal{S}^r$
and try to pick an appropriate solution $T_i$ from each class of solutions $(S_i,\ell_i)$, $i \in [1,k]$,
in such a way that the diversity
of the resulting tuple of
feedback vertex sets $(T_1,\ldots,T_r)$ is maximized.
The network of Section~\ref{sec:augmentation} must be
adapted to model the constraints resulting from solution classes.
Let $(S,\ell)$ be a solution class, with $|S|=b$.
For 
our construction, we just need to know
the family $\{\,\ell(v)\mid v\in S\,\} = 
\{L_1,L_2,\ldots,L_b\}$
of
disjoint nonempty vertex sets.
The solutions that are 
{described} by
this class are
all sets that can be obtained by picking
at least one
vertex from each set $L_q$.
Figure~\ref{fig:network2} shows the necessary adaptations for
\emph{one} solution $T=T_i$. In addition to a single node $T$
that is either directly of indirectly connected to all nodes $V_1,\ldots,V_n$,
like in Figure~\ref{fig:network},
we have additional
nodes 
representing the sets $L_q$.
For each vertex $j$ that appears in one of the sets $L_q$, there is an
additional node $U_j$ in an intermediate layer of the network.
The flow from $s$ to $L_q$ is
forced to be equal to 1, and this ensures that at least one element of
the set $L_q$ is chosen in the solution.
Here it is important
that the sets $L_q$ are disjoint.

A similar structure must be built
for each set $T_1,\ldots,T_r$,
and all these structures share the vertices $s$ and $V_1,\ldots,V_n$.
The rightmost layer of the network is
the same as in Figure~\ref{fig:network}.
\begin{figure}
  \centering
  \includegraphics {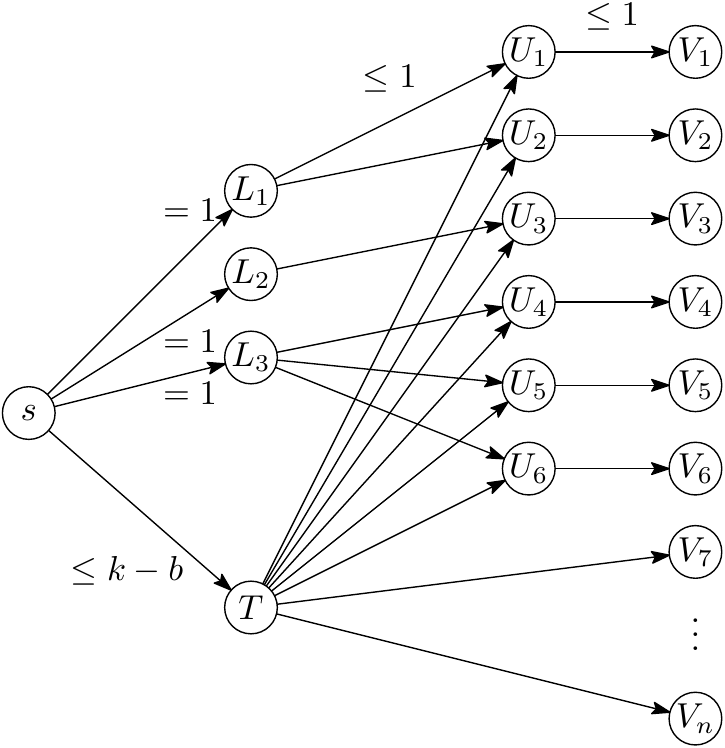}
  \caption{Part of the modified network for a solution $T$ which is specified by $b=3$ sets
    $L_1=\{1,2\}, L_2=\{3\}$, and $L_3=\{4,5,6\}$.}
  \label{fig:network2}
\end{figure}

The initial flow is not so straightforward as in
Section~\ref{sec:augmentation} but is still easy to find.
We simply saturate the arc from $s$ to  each of the nodes $L_q$ in turn
by a shortest augmenting path. Such a path can be found by a simple
reachability search in the residual network, in $O(rn)$ time.
The total running time $O(kr^2n)$ from
Section~\ref{sec:augmentation} remains unchanged.
\end{proof}

\section{Modeling Aspects: Discussion of the Objective Function}
\label{modeling}

In Sections~\ref{sec:framework} 
and~\ref{Diverse_FVS}, we have used 
the sum of the Hamming distances, $\sumdiversity$, as
the measure of diversity. While this metric is of natural interest, it appears
that in some specific cases, it may not be a useful choice. We present a simple
example where the \emph{most diverse} solution according to $\sumdiversity$
is not what one might expect.

Let $r$ be an even number.
We consider the path with $2r-2$
vertices, and
we are looking for $r$ vertex covers of size at most $r-1$,
of maximum diversity.
\begin{figure}[htb]
  \centering
  \includegraphics{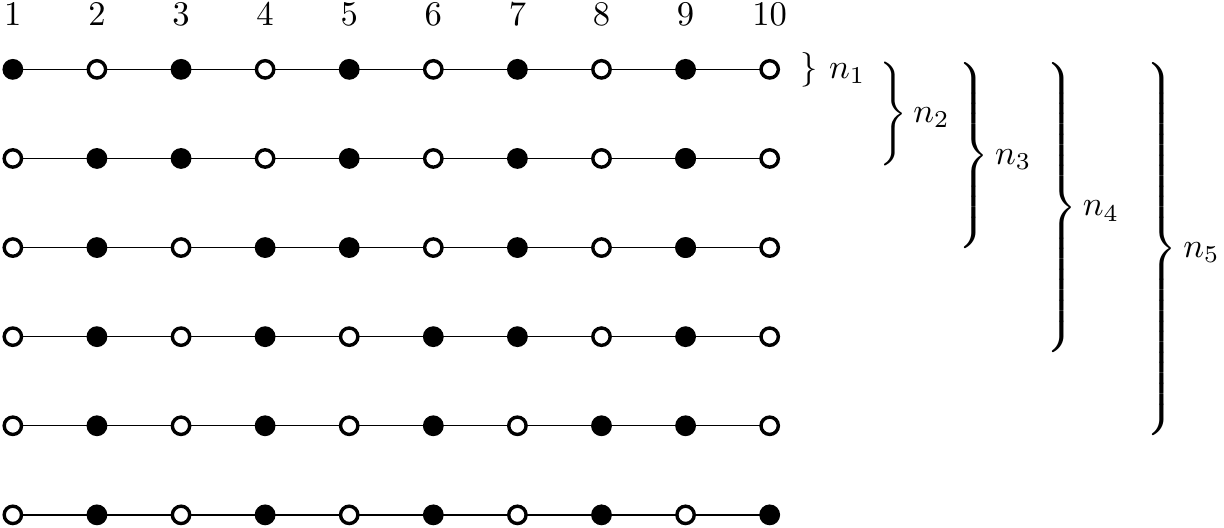}
  \caption{The $r=6$ different vertex covers of size $r-1=5$ in a path
    with $2(r-1)=10$ vertices}
  \label{fig:bad-example}
\end{figure}
Figure~\ref{fig:bad-example} shows an example with $r=6$.  The
smallest size of a vertex cover is indeed $r-1$, and there are $r$
different solutions. One would hope that the ``maximally diverse''
selection of $r$ solutions would pick all these different solutions.
But no, the selection that maximizes $\sumdiversity$ 
consists of
$r/2$ copies of just \emph{two} solutions, the ``odd''
vertices
and the
``even'' vertices
(the first and last solution in
Figure~\ref{fig:bad-example}).

This can be seen as follows.  If the
selected set contains in total $n_i$
copies of the first $i$ solutions in the order of
Figure~\ref{fig:bad-example},
then
the objective can be written as
\begin{displaymath}
  2 
  n_1(r-n_1)
  +  2n_2(r-n_2)
  + \cdots 
  +  2n_{r-1}(r-n_{r-1})
  .
\end{displaymath}
Here,
each term $2n_i(r-n_i)$ accounts for two consecutive vertices $2i-1,2i$
of the path in the formulation~\eqref{eq:obj-alternative}.
The unique way of  maximizing each
 term individually is to set $n_i=r/2$ for all~$i$.
This corresponds to the selection of $r/2$ copies of the first
solution and $r/2$ copies of the last solution, as claimed.

In a different setting, namely the distribution of $r$ points inside a square,
an analogous phenomenon has been observed
\cite[Figure~1]{defining-diversity-measures-2010}:
Maximizing the sum of pairwise Euclidean distances places all points
at the corners of the square.
In fact it is easy to see that, in this geometric setting, any locally
optimal
solution must place all
points 
on the boundary of the feasible region.
By contrast, for our combinatorial problem,
we don't know whether this pathological behavior is typical or
rare in instances that are not specially constructed.
Further research is needed.
A notion of diversity which
is more robust
in this respect is
the \emph{smallest}
difference between two solutions, which we consider in Section~\ref{max}.

\section{Maximizing the Smallest Hamming distance}
\label{max}
The undesired behavior highlighted in Section~\ref{modeling} is the fact that
the collection that maximizes the sum of the Hamming distances uses several
copies of the same set. In this section we explore how to handle this unexpected
behavior by changing the distance to the minimal Hamming distance between two
sets of the collection. This modification naturally removes the possibility of
selecting the same solution twice. We show how to solve \rMDkdHS and
\textsc{$r$-Min-Diverse $k$-Feedback Vertex Set} for this metric.

\begingroup
\def\thetheorem{\ref{thm:rMDkdHS_is_FPT}}
\begin{Theorem}
  \rMDkdHS can be solved in time
  \begin{itemize}
  \item $2^{kr^2}\cdot (kr)^{O(1)}$
    if $|U| < kr$ and
  \item \(d^{kr}\cdot |U|^{O(1)}\) otherwise.
  \end{itemize}
\end{Theorem}
\addtocounter{theorem}{-1}
\endgroup
\begin{proof}
  Let \((U, \F, k, r, t)\) be an instance of \rMDkdHS where $|U| = n$. If $n <
  kr$, we solve the problem by complete enumeration: There are trivially at most
  $2^n$ hitting sets of size at most $k$. We form all $r$-tuples $(T_1, \ldots,
  T_r)$ of them and select the one that maximizes $\mindiversity(T_1, \ldots,
  T_r)$. The running time is at most $O((2^n)^rr^2n)= O(2^{kr^2}kr^3)$.

  We now assume that $n \geq kr$.
  We use the same strategy as in Section~\ref{sec:framework}: We generate all
  $r$-tuples $(S_1, \ldots, S_r)$ of \emph{minimal solutions} and try to augment
  each one to a $r$-tuple $(T_1, \ldots, T_r)$ such that for each $i \in [1,r]$,
  $|T_i| \le k$ and $S_i \subseteq T_i\subseteq V(G)$ hold. The difference is that
  we try to maximize $\mindiversity(T_1, \ldots, T_r)$ instead of
  $\sumdiversity(T_1, \ldots, T_r)$ in the augmentation. Given that we have a
  large supply of $n\ge kr$ elements in $U$, this is easy. To each set $S_i$ we
  add $k-|S_i|$ new elements, taking care that we pick different elements for each
  $S_i$ with are not in any of the other sets $S_j$. The Hamming distance between
  two resulting sets is then
  $\Hammingdistance(T_i,T_j) = \Hammingdistance(S_i,S_j) + (k-|S_i|) + (k-|S_i|)$,
  and it is clear that this is the largest possibly distance that two sets
  $T'_i\supseteq S_i$ and $T'_j\supseteq S_j$ with $|T'_i|,|T'_j|\le k$ can
  achieve. Thus, since our choice of augmentation individually maximizes each
  pairwise Hamming distance, it also maximizes the smallest Hamming distance. This
  procedure can be carried out in $O(kr+n)=O(n)$ time. In addition, we need
  $O(kr^2)=O(n^2)$ time to compute the smallest distance.

  Using Lemma~\ref{lemma:minimaldHS}, we construct the set $\mathcal{S}$ of all
  minimal solutions of the \dHS instance \((U, \F)\), each of size at most \(k\).
  We then go through every $r$-tuple $(S_1, \ldots, S_r) \in \mathcal{S}^r$ and
  augment it optimally, as just described. The running time is
  $d^{kr} \cdot O(n^2)$.
\end{proof}

\begingroup
\def\thetheorem{\ref{thm:fvsminmax}}
\begin{Theorem}
  \rMDkFVS can be solved in time $2^{kr\cdot \max (r,7+\log_2(kr))}\cdot (nr)^{O(1)}$.
\end{Theorem}
\addtocounter{theorem}{-1}
\endgroup
\begin{proof}
  Let $G$ be a $n$-vertex graph.
  If $n < kr$, we again solve the problem by complete enumeration:
  There are trivially at most $2^n$
  feedback vertex sets of size at most $k$. We form all $r$-tuples
  $(T_1, \ldots, T_r)$ of them and
 select the one that maximizes
$\mindiversity(T_1, \ldots, T_r)$.
The running time is at most $O((2^n)^rr^2n)=
O(2^{kr^2} r^2n 
)$.

We assume now that $n \geq kr$.
As in Section~\ref{Diverse_FVS}, we 
construct a set $\mathcal{S}$ of at most $2^{7k}$ classes of FVS solutions of $G$,
using Lemma~\ref{lem:eqFVS}.
Then we go through all $(2^{7k})^r$ $r$-tuples of classes
$S = ((S_1,\ell_1), \ldots, (S_r,\ell_r)) \in \mathcal{S}^r$.
For each such $r$-tuple,
we look for the
$r$-tuple $(T_1, \ldots, T_r)$ of \fvs s such that each
 $T_i$ is described by $(S_i,\ell_i)$, and the objective value
 $\mindiversity(T_1, \ldots, T_r)$ is maximized.
 So far, the procedure is completely analogous to the algorithm of
 Theorem~\ref{thm:fvs}
 in 
Section~\ref{Diverse_FVS}
for maximizing $\sumdiversity(T_1, \ldots, T_r)$.
 
Now,
in going from a class $(S_i,\ell_i)$ to $T_i$, we have to select
a vertex from every set $\ell_i(v)$, for $v\in S_i$, and we may add an arbitrary
number of additional vertices, up to size~$k$.
We make this selection as follows:
Whenever $|\ell_i(v)|<kr$, we simply try all possibilities of choosing an
element of
$\ell_i(v)$ and putting it into $T_i$.
If $|\ell_i(v)|\ge kr$, we defer the choice for later.
In this way, we have created
at most $(kr)^{kr}$ ``partial'' \fvs s $(T^0_1, \ldots, T^0_r)$ 

For each such 
$(T^0_1, \ldots, T^0_r)$, we now add the remaining elements.
In each list $\ell_i(v)$ which has been deferred, we greedily pick an element
that is distinct from all other chosen elements. This is always possible since
the list is large enough. Finally, we fill up the sets to size $k$, again
choosing fresh elements each time. Each such choice is an optimal choice,
because it increases the Hamming distance between the concerned set $T_i$ and
\emph{every} other set $T_j$ by~1, which is the best that one can hope for. As
we proceed to this operation for each $S \in \mathcal{S}^r$, where
$|\mathcal{S}| \leq 2^{7k}$, and that for each such $S$, we create at most
$(kr)^{kr}$ $r$-tuples, we obtain an algorithm running in time
$2^{7kr}\cdot (kr)^{kr} \cdot n^{O(1)}$. The theorem follows.
\end{proof}

\section{Conclusions and Open Problems}\label{conclusion}

In this work, we have considered the paradigm of
finding small diverse collections of reasonably good solutions to combinatorial problems,
which has recently been introduced to the field of
fixed-parameter tractability theory~\cite{arBaFeJaOlRo2019}.

We have shown that finding diverse collections of \(d\)-hitting sets and
feedback vertex sets can be done in \FPT time. While these problems can be
classified as \FPT via the kernels and a treewidth-based meta-theorem proved
in~\cite{arBaFeJaOlRo2019}, the methods proposed here are of independent
interest. We introduced a method of generating a maximally diverse set of
solutions from a set that either contains all minimal solutions of bounded size
(\dHS) or from a collection of structures that in some way \emph{describes} all
solutions of bounded size (\textsc{Feedback Vertex Set}). In both cases, the
maximally diverse collection of solutions is obtained via a network flow model,
which does not rely on any specific properties of the studied problems.
It would be interesting to see if this strategy can be applied to give
FPT-algorithms for diverse problems that are not covered by the meta-theorem or
the kernels presented in~\cite{arBaFeJaOlRo2019}.

While the problems in~\cite{arBaFeJaOlRo2019} as well as the ones in
Sections~\ref{sec:framework} 
and~\ref{Diverse_FVS}, seek
to maximize the \emph{sum} of all pairwise Hamming distances, we also studied
the variant that asks to maximize the \emph{minimum} Hamming distance, taken
over each pair of solutions. This was motivated by an example where the former
measure does not perform as intended (Section~\ref{modeling}). We showed that
also under this objective, the diverse variants of \dHS and \FVS are \FPT. It
would be interesting to see whether this objective also allows for a (possibly
treewidth-based) meta-theorem.

In~\cite{arBaFeJaOlRo2019}, the authors ask whether there is a problem that is in \FPT parameterized by solution size 
whose $r$-diverse variant becomes \W{1}-hard upon adding $r$ as another component of the parameter. 
We restate this question here.
\begin{question}{\cite{arBaFeJaOlRo2019}}\label{q:w1}
Is there a problem $\Pi$ with solution size $k$, such that $\Pi$ is \FPT parameterized by $k$,
while \textsc{Diverse $\Pi$}, asking for $r$ solutions, is \W{1}-hard parameterized by $k + r$?
\end{question}

To the best of our knowledge, this problem is still wide open.
We believe that the $\mindiversity$ measure is more promising to obtain such a result rather than the $\sumdiversity$ measure.
A possible way to tackle both measures at once might be a parameterized (and strenghtened) analogue of the following approach that is well-studied in classical complexity.
Yato and Seta propose a framework~\cite{yato2003complexity} to prove \NP-completeness of finding a \emph{second} solution
to an \NP-complete problem.
In other words, there are some problems where given one solution it is still \NP-hard to 
determine whether the problem has a different solution.

From a different perspective, one might want to identify problems where obtaining one solution
is polynomial-time solvable, but finding a diverse collection of $r$ solutions becomes \NP-hard.
The targeted running time should be \FPT parameterized by $r$ (and maybe $t$, the diversity target) only.
We conjecture that this is most probably \NP- or \W-hard in general.
However, we believe it is interesting to search for well-known problems where it is not the case.

\paragraph{Acknowledgments.} The first, second, third and fourth authors would like to
    thank Mike Fellows for introducing them to the notion of
    diverse \FPT algorithms and sharing the manuscript ``The
    Diverse \(X\) Paradigm''~\cite{fellows2018diverseXparadigm}.

\bibliographystyle{plainurl}
\bibliography{diverse}

\end{document}